\documentclass[11pt,reqno]{amsart}
\usepackage[margin=3.6cm]{geometry}
\usepackage{graphicx}
\usepackage{amssymb}
\usepackage{bbm}

\newcommand{\R}{{\mathbb R}}

\newcommand{\C}{{\mathbb C}}

\newcommand{\ep}{\varepsilon}

\newtheorem{theorem}{Theorem}
\newtheorem{corollary}[theorem]{Corollary}
\newtheorem{proposition}[theorem]{Proposition}

\theoremstyle{definition}

\newtheorem{remark}{Remark}

\begin{document}

\title{On the topology of the zero sets of monochromatic random waves}

\author[Y. Canzani]{Yaiza Canzani}
\author[P. Sarnak]{Peter Sarnak}

\address[Y. Canzani]{ Institute for Advanced Study and  Harvard University.\medskip}
 \email{canzani@math.ias.edu}
\address[P. Sarnak]{ Institute for Advanced Study and  Princeton University.\medskip}
\email{sarnak@math.ias.edu}

\thanks{The research of Y.C. is partially supported by an NSERC Postdoctoral Fellowship and by NSF grant DMS-1128155.   The research of P.S. is supported by an NSF grant.}

\begin{abstract}
This note concerns the topology of the connected components of the zero sets of monochromatic random waves on compact Riemannian manifolds without boundary.
In \cite{SW} it is shown that these are distributed according to a universal measure on the space of smooth topological types. We determine the support of this measure.
\end{abstract}

\maketitle

%-------------------------------------------------------------------------------------------

\section{Introduction}

For $\ell \geq 0$ and $n \geq 2$ let $\mathcal E_\ell(S^n)$ denote the real linear space of (homogeneous) spherical harmonics of degree $\ell$ in $(n+1)$ variables. These are eigenfunctions of the Laplacian $\Delta_{S^{n}}$ on the sphere $S^n$ endowed with the round metric sitting in $\R^{n+1}$. With the corresponding $L^2$-inner product
 \[ \langle f, g \rangle = \int_{S^n} f(w) g(w) \, d\sigma(w),\]
 we get a Gaussian probability density $\mathbb P$ on $\mathcal E_\ell(S^n)$. Namely, 
 \[\mathbb P(A) = \int_A e^{- \langle f,f \rangle} df,\]
 where $df$ is the normalized Haar measure on  $\mathcal E_\ell(S^n)$ and $A \subset \mathcal E_\ell(S^n)$.
 
We are interested in the zero set $V(f)=f^{-1}(0)$ of a typical $f \in (\mathcal E_\ell(S^n), \mathbb P)$ as $\ell \to \infty$. Let $C(f)$ denote the connected components of $V(f)$. Then, for almost all $f$ these components are smooth, compact, $(n-1)$-dimensional manifolds.  The distribution of topologies of $V(f)$ is given by 
\begin{equation}\label{E: mu_f}
\mu_f:= \frac{1}{|C(f)|} \sum_{c \in C(f)} \delta_{t(c)},
\end{equation}
where $t(c)$ is the diffeomorphism type of $c$, and $\delta_\tau$ is the point measure at $\tau$. If we denote these types by $\widetilde{ H}(n-1)$  (it is a countable discrete set), then clearly $\mu_f$ is a probability measure on  $\widetilde{ H}(n-1)$. Let $H(n-1)$ denote the subset of $\widetilde{ H}(n-1)$  consisting of all those types that can be realized as embedded submanifolds of $\R^n$.

Nazarov and Sodin \cite{NS} have shown that for a typical $f\in \mathcal E_\ell(S^n)$ and $\ell \to +\infty$, $$|C(f)| \sim c_n\, \ell^n,$$ for some $c_n>0$. Since $C(f)$ consists of  many components,  it makes sense to examine the behavior of $\mu_f$ as $\ell \to +\infty$. In the recent work \cite{SW} it is shown that there is a probability measure $\mu_{\text{mono}}$ on $H(n-1)$ (denoted by $\mu_{C,n, 1}$ in \cite{SW})  such that for any $\ep>0$
\begin{equation}\label{E: prob discrepancy}
\mathbb P \left\{ f \in \mathcal E_\ell(S^n): \, D(\mu_f , \mu_{\text{mono}}) >\ep \right\} \to 0, 
\end{equation}
as $\ell \to +\infty$.  Here the discrepancy $D$ between two measures $\mu$ and $\nu$  is given by 
\begin{equation}\label{E: discrepancy}
D(\mu, \nu) =\sup \left\{ |\mu(F) - \nu(F)|:\;\; F \subset { \widetilde{ H}(n-1)}, \, F\,\, \text{is finite}\right\}.
\end{equation}

In fact, the same is proved for Gaussian ensembles of monochromatic waves on any given compact Riemannian manifold $(M,g)$ of dimension $n$ with no boundary (\cite{SW}). 
If $(M,g)$ is a compact Riemannian manifold, a monochromatic random wave of energy $T$ is defined  as the Gaussian ensemble of functions on $M$ given by 
$$f(x)=\sum_{T-\eta(T) \leq t_j \leq T} c_j \, \varphi_j(x).$$
Here, the functions $\varphi_j$ form an orthonormal basis of $L^2(M,g, \R)$  and are eigenfunctions of the Laplacian 
$\Delta_g \varphi_j +t_j^2 \varphi_j=0$.  The coefficients $c_j$ are independent Gaussian random variables of mean $0$ and variance $1$. Also, $\eta(T)=o(T)$ and $\eta(T)\to +\infty$ as $T\to +\infty$.
 
The measure $\mu_{\text{mono}}$ is the universal distribution for the topologies of the zero set of a typical monochromatic wave. Our aim in this note is to identify the support of this measure.
\begin{theorem} \label{T: support for mu in zero sets}
The support of $\mu_{\text{mono}}$ is equal to $H(n-1)$. That is, for all $c \in H(n-1)$,
$$\mu_{\text{mono}}(c)>0.$$
\end{theorem}

As shown in \cite{SW}, the following criterion, extending the condition $(\rho \, 4)$ of Sodin \cite[(1.2.2)]{Sod}, suffices to establish the above theorem:
Given $c \in H(n-1)$ find a trigonometric polynomial $f$ on $\R^n$ of the form
\begin{equation}\label{E: f}
f(x)=\sum_{\xi \in S^{n-1}} a_\xi \, e^{i \langle x, \xi \rangle},
\end{equation}
such that $f^{-1}(0)$ contains $c$ as one of its components. The sum in \eqref{E: f} is over a finite set of $\xi$'s, and the coefficients should satisfy $a_\xi=\overline{a_{-\xi}}$.\\

We note that for $n=2$, $H(1)$ is a point and so the statement of Theorem \ref{T: support for mu in zero sets} is trivial. For $n=3$, $H(2)=\{0,1,2,\dots\}$ with each $c \in H(2)$ being identified with its genus $g(c)$, and in this case  \eqref{E: f} is verified in \cite{SW} by deforming a carefully constructed $f$. For $n \geq 4$ the sets $H(n-1)$ are not known explicitly and we proceed here by more abstract and general arguments.

 First, we give a number of criteria which are equivalent to \eqref{E: f} and which are closely connected to the underlying translation invariant Gaussian field on $\R^n$.
The functions \eqref{E: f} satisfy 
\begin{equation}\label{E: eigenvalue 1}
(\Delta +1) f(x)=0 \qquad \text{on} \; \;\R^n,
\end{equation} 
and the various criteria reflect this (these f's being monochromatic!). 
We then apply some differential topology and Whitney's approximation Theorem to realize $c$ as an embedded real analytic submanifold of $\R^n$. Then, following some of the techniques in \cite{EP} we find suitable approximations of $f$ which satisfy \eqref{E: eigenvalue 1} and whose zero set contains a diffeomorphic copy of $c$.  The construction of $f$ hinges on the Lax-Malgrange Theorem and Thom's isotopy Theorem.

\subsection{Acknowledgement} We thank D.~Jakobson for pointing out to us the paper \cite{EP} by Enciso and Peralta-Salas, and I. Wigman for his valuable comments.

\section{Proof of Theorem \ref{T: support for mu in zero sets}}
Our interest is in the monochromatic Gaussian field on $\R^n$ which is a special case of the band limited Gaussian fields considered in \cite{SW}, and which is fundamental in the proof of \eqref{E: prob discrepancy} above. For $0 \leq \alpha \leq 1$, define the annulus $A_\alpha=\{ \xi \in \R^n: \; \alpha \leq |\xi| \leq 1\}$ and let $\sigma_\alpha$ be the Haar measure on $A_\alpha$ normalized so that $\sigma_\alpha(A_\alpha)=1$. Using that the transformation $\xi \mapsto -\xi$ preserves $A_\alpha$  we choose a real valued orthonormal basis $\{\phi_j\}_{j=1}^\infty$ of $L^2( A_\alpha, \sigma_\alpha)$ satisfying 
\begin{equation}\label{E: -xi}
\phi_j(-\xi)= (-1)^{\eta_j} \phi_j(\xi), \qquad \quad \eta_j \in \{0,1\}.
\end{equation}
The band limited Gaussian field $H_{n,\alpha}$ is defined to be the random real valued functions $f$ on $\R^n$ given by 
\begin{equation}\label{E: f, b}
f(x)=\sum_{j=1}^\infty b_j\, i^{\eta_j}\, \widehat {\phi_j}(x)
\end{equation}
where
\begin{equation}\label{E: hat phi}
\widehat {\phi_j}(x)= \int_{\R^n} \phi_j(\xi) e^{- i \langle x, \xi\rangle} d \sigma_\alpha(\xi)
\end{equation}
and the $b_j$'s are identically distributed, independent, real valued, standard Gaussian variables. We note that the field $H_{n, \alpha}$ does not depend on the choice of the orthonormal basis $\{\phi_j\}$.

The distributional identity $\sum_{j=1}^\infty  \phi_j(\xi)  \phi_j(\eta)=\delta(\xi-\eta)$ on $A_\alpha$ together with \eqref{E: -xi} lead to the explicit expression for the covariance function:
 \begin{equation}\label{E: covariance}
 \text{Cov}(x,y)=\mathbb E_{H_{n,\alpha}}(f(x)f(y))=\int_{\R^n} e^{i \langle x-y , \xi \rangle} d\sigma_\alpha(\xi).
 \end{equation}
From \eqref{E: covariance}, or directly from \eqref{E: f, b}, it follows that almost all $f$'s in $H_{n,\alpha}$ are analytic in $x$ \cite{AT}. For the monochromatic case $\alpha=1$ we have 
\begin{equation}\label{E: Covariance}
 \text{Cov}(x,y)=\frac{1}{(2\pi)^{\frac{n}{2}}}\frac{J_{\nu}(|x-y|)}{ |x-y|^{\nu}}, 
 \end{equation}
 where to ease notation we have set 
  \[\nu:=\frac{n-2}{2}.\]
 
 In this case there is also a natural choice of a basis for $L^2(S^{n-1},d\sigma)=L^2(A_1, \mu_1)$ given by spherical harmonics.  Let $\{Y^\ell_m\}_{m=1}^{d_\ell}$ be a real valued basis for the space of spherical harmonics  $\mathcal E_\ell(S^{n-1})$ of eigenvalue $\ell(\ell+n-2)$, where $d_\ell=\dim \mathcal E_\ell(S^{n-1})$. 
We compute the Fourier transforms for the elements of this basis.
\begin{proposition} For every  $\ell \geq 0$ and $m=1, \dots, d_\ell$, we have 
\begin{equation}\label{E: f.t.  harmonic}
 \widehat{Y^\ell_m}(x)=(2\pi)^{\frac{n}{2}}\,i^\ell\, Y^\ell_m \left(\frac{x}{|x|}\right) \frac{J_{\ell + \nu}(|x|)}{|x|^{\nu}}.
\end{equation}
\end{proposition}

\begin{proof}
We give a proof using the theory of point pair invariants \cite{Sel} which places such calculations in a general and conceptual setting. The sphere $S^{n-1}$ with its round metric is a rank $1$ symmetric space and $\langle \dot x, \dot y \rangle$ for $\dot x, \dot y \in S^{n-1}$ is a point pair invariant (here $\langle \, , \, \rangle$ is the standard inner product on $\R^n$ restricted to $S^{n-1}$). Hence, by the theory of these pairs we know that for every function $h:\R \to \C$ we have
\begin{equation}\label{E: int=lambda}
\int_{S^{n-1}} h(\langle \dot x ,\dot y \rangle ) \,Y (\dot y) \,d\sigma(\dot y)= \lambda_h (\ell) Y (\dot x), 
\end{equation}
where $Y$ is any spherical harmonic of degree $\ell$ and $\lambda_f(\ell)$ is the spherical transform. The latter can be computed explicitly using the zonal spherical function of degree $\ell$.
Fix any $\dot x \in S^{n-1}$  and  let $Z^\ell_{\dot x}$ be the unique spherical harmonic of degree $\ell$ which is rotationally invariant by motions of $S^{n-1}$ fixing $\dot x$ and so that $Z^\ell_{\dot x}(\dot x)=1$. Then, 
\begin{equation}\label{E: lambda}
\lambda_h(\ell)=\int_{S^{n-1}} h(\langle \dot x, \dot y \rangle) Z^\ell_{\dot x}(\dot y)\, d\sigma(\dot y).
\end{equation}
The function $Z^\ell_{\dot x}(\dot y)$ may be expressed in terms of the Gegenbauer polynomials  \cite[(8.930)]{GR} as
\begin{equation}\label{E: zonal, gegenb}
 Z^\ell_{\dot x}(\dot y)=\frac{C_\ell^\nu \big(\big \langle \dot x, \dot y \big \rangle \big) }{C_\ell^{\nu}(1) }.
\end{equation}
Now, for $x \in \R^n$,
\[ \widehat{Y^\ell_m}(x)=\int_{S^{n-1}} h_x\big( \big \langle \tfrac{x}{|x|}, \dot y  \big \rangle \big)\, Y^\ell_m(\dot y) d\sigma(\dot y),\]
where we have set $h_x(t)=e^{-i|x| t}$. Hence, by \eqref{E: int=lambda} we have
\[ \widehat{Y^\ell_m}(x)= \lambda_{h_x} (\ell) \, Y^\ell_m\big( \tfrac{x}{|x|} \big),\]
with 
\begin{equation}\label{E: computation of lambda}
\lambda_{h_x}(\ell)
= \int_{S^{n-1}} e^{-i |x| \big \langle  \tfrac{x}{|x|}\, ,\, \dot y \big \rangle }Z^\ell_{\dot x}(\dot y)\, d\sigma(\dot y)
= \frac{\text{vol}(S^{n-2})}{C_\ell^\nu(1) } \int_{-1}^1 e^{-it|x|} \, C^\nu_\ell(t) (1-t^2)^{\nu -\frac{1}{2}} \, dt.
\end{equation}
The last term in \eqref{E: computation of lambda} can be computed using   \cite[(7.321)]{GR}. This gives
\[\lambda_{h_x}(\ell)=(2\pi)^{\frac{n}{2}}\, i^\ell\, \frac{J_{\ell+\nu}(|x|)}{|x|^\nu},\]
as desired.

% Another -less insightful- way to verify the result is to use  \cite[(1.2.8)]{DX}  
% combined with the summation formula \cite[(8.532.1)]{GR} to get
%\begin{align*}
%\E(f(x)f(y))
%=\frac{1}{\nu} \sum_{\ell=0}^\infty  (\ell +\nu)  \frac{J_{\ell+ \nu}(|x|)}{|x|^{\nu} } \frac{J_{\ell+ \nu}(|y|)}{|y|^{\nu}} C_\ell^\nu \left(\frac{\langle x, y \rangle}{|x||y|}\right)
%=\frac{1}{2^{\nu} \nu \Gamma(\nu)} \frac{J_{\nu}(|x-y|)}{|x-y|^\nu}.
%\end{align*}
%The result follows from observing that  $\int_{S^{n-1}}e^{i\langle x-y, w \rangle} \text{dvol}(w)=\frac{1}{(2\pi)^{n/2}}  \frac{J_{\nu}(|x-y|)}{|x-y|^\nu}.$
\end{proof} 

\begin{corollary}
The monochromatic Gaussian ensemble $H_{n,1}$ is given by random $f$'s  of the form
\begin{equation*}
{f}(x)=(2\pi)^{\frac{n}{2}} \sum_{\ell=0}^\infty \sum_{m=1}^{d_{\ell}} b_{\ell,m} \,  Y^\ell_m \left(\frac{x}{|x|}\right) \frac{J_{\ell + \nu}(|x|)}{|x|^{\nu}},
\end{equation*}
where the $b_{\ell, m}$'s are i.i.d  standard  Gaussian variables.
\end{corollary}

The functions $x \mapsto Y^\ell_m \left(\frac{x}{|x|}\right) \frac{J_{\ell + \nu}(|x|)}{|x|^{\nu}}$, $x \mapsto e^{i \langle x, \xi\rangle}$ with $|\xi|=1$, and those in \eqref{E: f, b} for which the series converges rapidly (eg. for almost all $f$ in $H_{n,1}$), all satisfy \eqref{E: eigenvalue 1}.
We therefore introduce the space 
\[E_1:=\{f : \R^n \to \R: \; f \in \text{Ker} (\Delta +1)\}.\]
In addition, consider the subspaces $P_1$ and $T_1$ of $E_1$ defined by
\[P_1:= \text{span}\left\{x \mapsto Y^\ell_m \left(\frac{x}{|x|}\right) \frac{J_{\ell + \nu}(|x|)}{|x|^{\nu}}:\; \, \ell \geq 0, \; m=1, \dots, d_\ell \right\},\] 
\[T_1:= \text{span}\left\{x\mapsto\frac{e^{i \langle x, \xi\rangle} +e^{-i \langle x, \xi\rangle} }{2}\, ,\,x\mapsto \frac{e^{i \langle x, \xi\rangle} -e^{-i \langle x, \xi\rangle} }{2i}:\;\; |\xi|=1 \right\} .\] 

\begin{proposition}\label{E: approximation}
Let  $f \in E_1$ and let $K \subset \R^n$ be a compact set. Then, for any $t \geq 0$ and $\ep>0$ there are $g \in P_1$ and $h \in T_1$ such that 
\[\|f-g\|_{C^t(K)}<\ep \qquad \text{and}\qquad \|f-h\|_{C^t(K)}<\ep.\]
That is, we can approximate $f$ on compact subsets in the $C^t$-topology  by elements of $P_1$ and $T_1$ respectively. 
\end{proposition}

\begin{proof}
Let $f\in E_1$. Since $f$ is analytic we can expand it in a rapidly convergent series in the $Y^\ell_m$'s. That is, 
\[f(x)=\sum_{\ell=0}^\infty \sum_{m=1}^{d_\ell} a_{m,\ell}(|x|)Y^\ell_m( \tfrac{x}{|x|}).\]
Moreover, for $r>0$,
\begin{equation}\label{E: fourier a's}
\int_{S^{n-1}}|f(r \dot x)|^2\, d\sigma(\dot x)= \sum_{\ell=0}^\infty \sum_{m=1}^{d_\ell} |a_{m,\ell}(r)|^2.
\end{equation}
In polar coordinates, $(r, \theta) \in (0, +\infty)\times S^{n-1}$,   the Laplace operator in $\R^n$  is given by  
\[\Delta= \partial_r^2 + \frac{n-1}{r} \partial_r  +\frac{1}{r^2} \Delta_{S^{n-1}},\] 
and hence for each $\ell,m$ we have that 
 \begin{equation}\label{E: a}
 r^2a_{m,\ell}''(r) + (n-1) r a_{m,\ell}'(r) +(r^2- \ell(\ell+n-2))a_{m,\ell}(r)=0.
 \end{equation}
where $\ell$ is some positive integer.
There are two linearly independent solutions to \eqref{E: a}. One is $r^{-\nu}J_{\ell+\nu}(r)$ and the other blows up as $r \to 0$. Since the left hand side of \eqref{E: fourier a's} is finite as $r\to 0$, it follows that the $ a_{m,\ell}$'s cannot pick up any component of the blowing up solution. That is, for $r \geq 0$
\[ a_{m,\ell}(r)= c_{\ell,m}  \frac{J_{\ell +\nu}(r)}{r^\nu},\]
for some $c_{m,\ell} \in \R$.
Hence, 
\begin{equation}\label{E:  expansion for f}
{f}(x)=\sum_{\ell=0}^\infty \sum_{m=1}^{d_{\ell}} c_{\ell,m} \; Y^\ell_m \left(\frac{x}{|x|}\right) \frac{J_{\ell + \nu}(|x|)}{|x|^{\nu}}.
\end{equation}
Furthermore, this series converges absolutely and uniformly on compact subsets, as also do its derivatives. Thus, $f$ can be approximated by members of $P_1$ as claimed, by simply truncating the series in \eqref{E:  expansion for f}. 

To deduce the same for $T_1$ it suffices to approximate each fixed $Y^\ell_m \left(\frac{x}{|x|}\right) \frac{J_{\ell + \nu}(|x|)}{|x|^{\nu}}$. To this end let $\xi_1, -\xi_1, \xi_2, -\xi_2,\dots ,\xi_N, -\xi_N$ be a sequence of points in $S^{n-1}$ which become equidistributed with respect to $d\sigma$  as $N \to \infty$.
Then, as $N\to \infty$,
\begin{equation}\label{E: limit}
 \frac{1}{2N} \sum_{j=1}^N \left( e^{-i \langle x, \xi_j \rangle} Y^\ell_m(\xi_j)+ (-1)^\ell e^{i \langle x, \xi_j \rangle} Y^\ell_m(\xi_j)\right) \longrightarrow \int_{S^{n-1}} e^{-i \langle x, \xi \rangle} Y^\ell_m(\xi) \,d\sigma(\xi).
 \end{equation}
The proof follows since $(2\pi)^{\frac{n}{2}}\, i^\ell\, Y^\ell_m \left(\frac{x}{|x|}\right) \frac{J_{\ell + \nu}(|x|)}{|x|^{\nu}}= \int_{S^{n-1}} e^{-i \langle x, \xi \rangle} Y^\ell_m(\xi) \,d\sigma(\xi)$. Indeed, the convergence in \eqref{E: limit} is uniform over compact subsets in $x$. 
\end{proof}

\begin{remark}
For $\Omega \subset \R^n$ open, let $E_1(\Omega)$ denote the eigenfunctions on $\Omega$ satisfying $\Delta f(x)+f(x)=0$ for $x \in \Omega$. Any function $g$ on $\Omega$ which is a limit (uniform over compact subsets of $\Omega$) of members of $E_1$ must be in $E_1(\Omega)$. While the converse is not true in general, note that if $\Omega=B$ is a ball in $\R^n$, then the proof of Proposition \ref{E: approximation} shows that the uniform limits of members of $E_1$ (or $P_1$, or $T_1$) on compact subsets in $B$ is precisely $E_1(B)$. 
\end{remark}

With these equivalent means of approximating functions by suitable members of $H_{n, 1}$ we are ready to prove Theorem \ref{T: support for mu in zero sets}. To verify the criterion following Theorem \ref{T: support for mu in zero sets} we can use any   function in $E_1$.

\subsection{Proof of Theorem \ref{T: support for mu in zero sets}}

By the discussion  above it follows that given  a representative $c$ of a class $t(c) \in H(n-1)$,  it suffices to find $f \in E_1$  for which $C(f)$ contains a diffeomorphic copy of $c$. 

To begin the proof we claim that we may assume that $c$ is real analytic. Indeed, if we start with $\tilde{c}$ smooth, of the desired topological type, we may construct a tubular neighbourhood $V_{\tilde c}$ of $\tilde {c}$  and a smooth  function 
$$H_{\tilde c}:V_{\tilde c} \to \R \qquad \text{with}\qquad \tilde c= H_{\tilde c}^{-1}(0).$$ Note that without loss of generality we may assume that   $\inf_{x \in V_{\tilde c}}\|\nabla H_{\tilde c}(x)\|>0$.
Fix any $\epsilon>0$. We apply Thom's isotopy Theorem \cite[Thm 20.2]{AR} to obtain the existence of a constant $\delta_{\tilde c}>0$ 
so that for any function $F$ with $\|F-H_{\tilde c}\|_{C^1(V_{\tilde c})} <\delta_{\tilde c}$ there exists $\Psi_{F} :\R^n \to \R^n$ diffeomorphism with $\|\Psi_{F}-Id\|_{C^0(\R^n)} <\ep$, $\text{supp}(\Psi_{F}-Id)  \subset V_{\tilde c}$, and  
$$\Psi_{F}(\tilde c)= F^{-1}(0) \cap V_{\tilde c}.$$
To construct a suitable $F$ we use Whitney's approximation Theorem \cite[Lemma 6]{Wh} which yields the existence of a real analytic approximation $F:V_{\tilde c} \to \R^{m_{\tilde c}}$  of $H_{\tilde c}$ that satisfies $\|F-H_{\tilde c}\|_{C^1(V_{\tilde c})}< \delta_{\tilde c}$. It follows that $\tilde c$ is diffeomorphic to $c:=\Psi_{F}(\tilde c)$ and $c$ is real analytic as desired. 

By the Jordan-Brouwer Separation Theorem \cite{Li}, the hypersurface $c$ separates $\R^n$ into two connected components. We write $A_c$ for the corresponding bounded component of $\R^n \backslash c$.
Let $\lambda^2$ be the first Dirichlet eigenvalue for the domain $A_c$ and let $h_\lambda$ be the corresponding eigenfunction: 
$$\begin{cases}
(\Delta + \lambda^2 )h_\lambda(x)=0& x \in \overline{A_c},\\
 h_\lambda(x)=0& x \in c.
\end{cases}$$
Consider the rescaled function $$h(x):=h_\lambda(x/\lambda),$$ defined on the rescaled domain $\lambda A_c:=\{x \in \R^n:  x/\lambda \in A_c\}$. 
Since $(\Delta + 1)h=0$ in $\overline{\lambda A_c}$, and $\partial (\lambda A_c)$ is real analytic, $h$ may be extended to some open set $B_c\subset \R^n$ with $\overline{\lambda A_c} \subset B_c$ so that
$$\begin{cases}
(\Delta + 1)h(x)=0& x \in B_c,\\
 h(x)=0& x \in  \lambda c,
\end{cases}$$
where $\lambda c$ is the rescaled hypersurface $\lambda c:= \{x \in \R^n:\; x/\lambda \in c\}$.
Note that since $h_\lambda$ is the first Dirichlet eigenfunction, then we know that there exists a tubular neighbourhood $V_{\lambda c}$ of $\lambda c$ on which  $\inf_{x \in V_{\lambda c}}\|\nabla h(x)\|>0$ (see Lemma 3.1 in \cite{BHM}). Without loss of generality assume that  $V_{\lambda c}\subset B_c$.

Given any $\ep>0$ we apply Thom's isotopy Theorem \cite[Thm 20.2]{AR} to obtain the existence of a constant $\delta>0$ 
so that for any function $f$ with $\|f-h\|_{C^1(V_{\lambda c})} <\delta$ there exists $\Psi_f :\R^n \to \R^n$ diffeomorphism so that $\|\Psi_f-Id\|_{C^0(\R^n)} <\ep$, $\text{supp}(\Psi_f-Id)  \subset V_{\lambda c}$, and  
$$\Psi_f(\lambda c)= f^{-1}(0) \cap V_{\lambda c}.$$
Since $\R^n \backslash B_c$ has no compact components,  Lax-Malgrange's Theorem \cite[p. 549]{Kr} yields the existence of a global solution $f:\R^n \to \R$ to the elliptic equation $(\Delta + 1)f=0$ in $\R^n$ with $$\|f-h\|_{C^1(B_c)}<\delta.$$
We have then constructed a solution to $(\Delta + 1)f=0$ in $\R^n$, i.e. $f\in E_1$,  for which $f^{-1}(0)$ contains a diffeomorphic copy of $c$ (namely, $\Psi_f(\lambda c))$. This concludes the proof of the theorem.

\qed

 We note that the problem of finding a solution  to $(\Delta+1)f=0$ for which $C(f)$ contains a diffeomorphic copy of $c$ is related to the work \cite{EP} of A.~ Enciso and D.~Peralta-Salas. In \cite{EP} the authors seek to find solutions to the problem $(\Delta -q)f=0$ in $\R^n$ so that $C(f)$  contains a diffeomorphic copy of $c$, where $q$ is a {nonnegative}, real analytic, potential and $c$ is a (possibly infinite) collection of compact or unbounded ``tentacled" hypersurfaces. The construction of the solution $f$ that we presented is based on the ideas used in \cite{EP}.  Since our setting and goals are simpler than theirs, the construction  of $f$  is much shorter and straightforward.

%--------------------------------------------------------------------------------

\end{document}